\documentclass[reqno]{amsart}
\usepackage{graphicx}

\newtheorem{lemma}{Lemma}
\newtheorem{theorem}{Theorem}

\newtheorem{remark}{Remark}

\newcommand{\bee}{\begin{eqnarray}}
\newcommand{\eee}{\end{eqnarray}}
\newcommand{\be}{\begin{eqnarray*}}
\newcommand{\ee}{\end{eqnarray*}}
\newcommand{\R}{{\mathbb R}}
\newcommand{\N}{{\mathbb N}}
\newcommand{\Z}{{\mathbb Z}}

\newcommand{\Ss}{{\Gamma}}

\newcommand{\St}{s}

\begin{document}
 
 \title [Spectral Splitting for NLS]{Spectral splitting method for nonlinear Schr\"odinger equations with quadratic potential}
 
 \author {
 Andrea Sacchetti
 }

\address {
Department of Physics, Informatics and Mathematics, University of Modena and Reggio Emilia, Modena, Italy.
}

\email {andrea.sacchetti@unimore.it}

\date {\today}

\thanks {This work is partially supported by GNFM-INdAM and by the UniMoRe-FIM project ``Modelli e metodi della Fisica Matematica''.}

\begin {abstract} In this paper we propose a modified Lie-type spectral splitting approximation where the external potential is of quadratic type. \ It is proved 
that we can approximate the solution to a one-dimensional nonlinear Schr\"odinger equation by solving the linear problem and treating the nonlinear term separately, 
with a rigorous estimate of the remainder term. \ Furthermore, we show by means of numerical experiments that such a modified approximation is more efficient than 
the standard one.
\end{abstract}

\keywords {Spectral splitting approximation; nonlinear Schr\"odinger equation; harmonic and inverted potential; evolution operator.}

\maketitle

\section {Introduction} \label {Sez1}

In this paper we consider non-linear Schr\"odinger equation of the form
\bee
\left \{ 
\begin {array}{l}
i \hbar \frac {\partial \psi_t (x)}{\partial t} = \left [ - \frac {\hbar^2}{2m} \Delta + V (x) \right ] \psi_t (x) + \nu |\psi_t (x) |^{2\sigma } \psi_t (x)  \\
\psi_{t_0} (x)= \psi_0 (x) 
\end {array}
\right. \, ,\psi_t (\cdot ) \in L^2 (\R^d , dx) \, ,   \label {eq1}
\eee
where $\sigma>0$, $V(x)$ is a real-valued quadratic potential and $\nu \in \R$. \ Hereafter, we assume the units such that $2m=1$ and $\hbar=1$, we simply 
denote by $\psi_t$ the wavefunction $\psi_t (x)$, by $\psi_0$ the initial wavefunction $\psi_0 (x)$, $\psi' = \frac {\partial \psi}{\partial x}$, 
$\psi'' = \frac {\partial^2 \psi}{\partial x^2}$, etc., and $\dot \psi = \frac {\partial \psi}{\partial t}$. \ Furthermore, we restrict 
our attention, for sake of simplicity, to the one-dimensional case, i.e. $d=1$. 

Nonlinear Schr\"odinger equations with a quadratic potential are a useful tool in order to describe Bose-Einstein condensates in a trapping potential 
\cite {DGPS,Sal}, as well as in the theory of nonlinear optics \cite {NonLinOpt}.

An efficient numerical treatment of such an equation is based on the Lie-type splitting approximation. \ The basic idea is quite simple {(see, e.g., the paper \cite {Besse})}: suppose to consider an 
evolution equation
\bee
\left \{ 
\begin {array}{l}
i \dot \psi_t = \left [ A + B \right ] \psi_t   \\
\psi_{t_0} = \psi_0 
\end {array}
\right. \, ,\psi_t \in L^2 (\R , dx) \, ,   \label {eq2}
\eee
where $A$ and $B$ are two given operators. \ Let us denote by $S^{t-t_0} \psi_0$ the solution to (\ref {eq2}) where $S^{t-t_0}$ is the associated evolution operator; 
let us denote by $X^{t-t_0}$ and $Y^{t-t_0}$ the evolution operators respectively associated to the equations
\be
i \dot \psi_t = A \psi_t \ \mbox { and } \ i \dot \psi_t = B \psi_t \, . 
\ee
It is well known that, in general, 
\be
S^{\delta} \psi_0 \not= X^{\delta} Y^{\delta} \psi_0\, , \ \delta \in \R \, , 
\ee
but this difference may be proved, under some circumstances, to be small when $\delta$ is small. \ More precisely, if one fix any $T>0$, a $\delta >0$ small enough 
and a positive integer number $n$ such that $n\delta \le T$, then the solution $\psi_t =S^{t-t_0} \psi_0$ to (\ref {eq2}), where $t= n \delta +t_0$, can be 
approximated by 
\bee
\left [ X^\delta Y^\delta \right ]^n \psi_0 \, , \label {eq3}
\eee
up to a remainder term that goes to zero when $\delta$ goes to zero.

In fact, a better result may be obtained by means of the Strang-type approximation where the solution $\psi_t$ to (\ref {eq2}) is approximated by
\be
\left [ X^{\delta /2} Y^\delta X^{\delta /2} \right ]^n \psi_0 \, . 
\ee
However, for sake of definiteness we restrict our analysis to the Lie-type approximation method (\ref {eq3}).

When one applies such an approximation to the problem (\ref  {eq1}) a typical choice consists in choosing $A=- \frac {\partial^2}{\partial x^2}$, i.e. the 
one-dimensional linear Laplacian operator, and $B= V + \nu |\psi_t |^{2\sigma}$. \ Here, we denote by $X_1^\delta$ and $Y_1^\delta$ the associated evolution 
operators. \ Thus, with such a choice $X_1^{t-t_0} = e^{-iA(t-t_0)}$ is the evolution operator associated to the Laplacian and it is an 
integral {linear} operator with well known kernel function. \ For what concerns $Y_1^{t-t_0}$ it is the evolution operator obtained by means of 
the solution to the ordinary differential equation
\bee
\left \{ 
\begin {array}{l}
i  \dot w_t = V w_t + \nu |w_t |^{2\sigma} w_t \\ 
w_{t_0} =w_0  
\end {array}
\right. \, . 
\label {eq4}
\eee
We observe that $|w_t |$ is constant with respect to $t$ since $V(x)$ is a real-valued function; indeed, one can check that
\be
\frac {\partial |w_t |^2}{\partial t} &=&  \frac {\partial w_t }{\partial t} \overline {w_t } + 
\frac {\partial \overline {w_t  }}{\partial t} {w_t } \\ 
&=& {-i} \left [  V w_t + \nu |w_t|^{2\sigma} w_t \right ] \bar w_t + {i} \left [  V \bar w_t + \nu |w_t|^{2\sigma} \bar w_t \right ]  w_t =0 \, . 
\ee
Thus, equation (\ref {eq4}) takes the form
\bee
\left \{ 
\begin {array}{l}
i  \dot w_t = \left [ V + \nu |w_0 |^{2\sigma} \right ] w_t \\ 
w_{t_0}=w_0 
\end {array}
\right.
\label {eq5}
\eee
which has solution
\bee
w_t(x)= \left [ Y_1^{t-t_0} w_0 \right ] (x) = e^{-i [V(x)+ \nu |w_0 (x)|^{2\sigma} ](t-t_0)} w_0 (x) \, , \label {eq5Bis}
\eee
that is $Y_1^{t-t_0}$ is a multiplication {nonlinear} operator such that
\be
\| Y^{t-t_0} w_0 \|_{L^p} = \| w_0 \|_{L^p} \, , \ \forall p \in [1,+\infty ] \,  .
\ee

Therefore, both evolution operators $X_1^{t-t_0}$ and $Y_1^{t-t_0}$ have an explicit expression.

The crucial point is to give a rigorous estimate of the remaining term 
\bee
{\mathcal R}_1 \psi_0:= S^{t-t_0}\psi_0 - \left [ X_1^\delta Y_1^\delta \right ]^n \psi_0 \, . \label {eq6}
\eee

Let us recall here some rigorous results concerning the estimate of ${\mathcal R}_1$. 

In the case where the external potential is absent, i.e. $V \equiv 0$, and under some assumption on the initial state $\psi_0 $ then the estimate 
\bee
\| {\mathcal R}_1 \psi_0 \|_{L^2} \le C {\delta}\, ,  \label {eq7}
\eee
for some positive constant $C=C(\psi_0 ,T)$, has been proved by \cite {Besse,CG}.

If the external potential $V$ is not identically zero then a similar estimate of the remainder term holds true provided that the Schr\"odinger equation is 
restricted to a bounded domain $U \subset \R^d$  and provided that its solution $\psi_t $ is such that (see, e.g. Thm. 4.3 \cite {BC})
\be
\psi \in C\left ( [0,T]; H^m (U) \cap H_0^1 (U) \right )\, ,
\ee
for some $m \ge 5$.

We should also mention that a purely formal (not completely rigorous) argument (see \cite {AKKT}) suggests that
\be
\| {\mathcal R}_1 \psi_0 \|_{L^2 (\R )} \le C \delta^2 e^{C\delta } 
\ee
for some positive constant $C=C(\psi_0 , T)$, provided that the potential $V(x)$ is a bounded function and $\psi_0 \in H^2 (\R)$.

We must remark that such an approach does not properly work when the potential $V(x)$ is singular, e.g. $V$ is a Dirac's delta. \ In such a case the method should 
be modified by choosing $A=H = - \frac {\partial^2}{\partial x^2} + V$, where $H$ is the linear Schr\"odinger operator, and where $B= \nu |\psi |^{2\sigma}$ is the 
nonlinear term \cite {S}.

In this paper we prove the validity of the Lie-type approximation for a nonlinear Schr\"odinger equation with quadratic potential following the approach introduced 
by \cite {S} in the case of singular potential. \ Let $X_2^\delta := e^{-i H\delta}$ be the evolution operator associated to the 
linear Schr\"odinger operator and 
\bee 
\left [ Y_2^\delta w \right ] (x) := e^{-i \nu |w (x)|^{2\sigma } \delta } w(x) \, . \label {eq8}
\eee 
If we denote by
\bee
{\mathcal R}_2 \psi_0:= S^{t-t_0}\psi_0 - \left [ X_2^\delta Y_2^\delta \right ]^n \psi_0  \label {eq9}
\eee
the remainder term, we are going to prove that it goes to zero when $\delta$ goes to zero and $n \delta \le T$ for any fixed $T>0$ (see Theorem \ref {Teo_uno}). \ We can 
thus show that this method has at least as solid a theoretical basis as the one based on the approximation (\ref {eq6}).

One must remark that approximation (\ref {eq6}) can be implemented by means of a quite simple numerical algorithm  
basically independent on the shape of the potential $V(x)$; in contrast, approximation (\ref {eq9}) is substantially useful when the evolution operator 
$X_2^\delta$, associated to the linear Schr\"odinger operator, can be efficiently computed, like in the case of a quadratic potential. \ On the other side, by means 
of numerical experiments, the approximation (\ref {eq9}) turns out to be more accurate than the usual one (\ref {eq6}).

The paper is organized as follows. \ In Section \ref {Sez2} we state our main result (Theorem \ref {Teo_uno}); Section \ref {Sez3} is devoted to the proof of 
Theorem \ref {Teo_uno}; in Section \ref {Sez4} we compare the approximations (\ref {eq6}) and (\ref {eq9}) on test models; is Section \ref {Sez5} we draw the 
conclusions; a short Section \ref {SezA} appendix is devoted to the Mehler's formulas, that is to the kernel of the evolution operator $X_2^\delta$ of the linear Schr\"odinger operator with harmonic or 
inverted oscillator potential.

Hereafter $C$ denotes any positive constant which may change from line to line.

\section {Main result} \label {Sez2}

Let us consider the one-dimensional {(i.e. $d=1$)} nonlinear Schr\"odinger equation of the form
\bee
\left \{ 
\begin {array}{l}
i  \frac {\partial \psi_t}{\partial t} = H\psi_t + \nu |\psi_t |^{2\sigma } \psi_t   \\
\psi_{t_0} = \psi_0  
\end {array}
\right. \, ,\psi_t \in L^2 (\R , dx) \, , \ H = -  \frac {\partial^2 }{\partial x^2} + V (x)\, ,  \label {eq10}
\eee
where 
\be
V(x)= \alpha x^2 
\ee
is a real-valued quadratic potential for some $\alpha \in \R \setminus \{ 0\}$. \ Let $t_0=0$ for the sake of definiteness.

Solutions to (\ref {eq10}) are usually studied in the space
\be
\Sigma := \left \{ \psi \in {\mathcal S}' \ : \ \| \psi \|_{\Sigma} := \| \psi \|_{L^2} + \left \| \psi' \right \|_{L^2} + \| x \psi \|_{L^2} < + \infty \right \} 
\ee
and the existence of a local solution to (\ref {eq10}), with the conservation of the norm
\be
{\mathcal N}(\psi_t )= {\mathcal N}(\psi_0 ) \, , \  \mbox { where } {\mathcal N}(\psi ):= \| \psi \|_{L^2} \, , 
\ee
and of the energy
\be
{\mathcal E}(\psi_t )= {\mathcal E}(\psi_0 ) \, ,\ \mbox { where } \ {\mathcal E}(\psi ):= \left \| \psi' \right \|_{L^2}^2 + \alpha \| x \psi \|_{L^2}^2 + 
\frac {\nu}{\sigma +1} \| \psi \|_{L^{2\sigma +2}}^{2\sigma +2} \, . 
\ee
has been proved (see \cite {C1,C2}).

Solution to (\ref {eq10}) globally exists when $\sigma < \frac {2}{d}$ and the map $t\in \R \to \psi_t \in \Sigma$ is continuous provided that 
$\psi_0 \in \Sigma$. \ On the other hand, blow-up may occur as proved by \cite {C2} under some circumstances for some $\nu <0$ and $\alpha >0$ when $\sigma \ge \frac d2$. 

Let ${\Ss} $ be the vector space
\be
{\Ss} = \left \{ \psi \in {\mathcal S}' \ : \ \| \psi \|_{\Ss} := \| \psi \|_{H^2} + \| x^2 \psi \|_{L^2} < + \infty \right \}\subset \Sigma \, . 
\ee
Let $X^\delta_2 = e^{-iH\delta}$ be the evolution operator associated to the linear Schr\"odinger operator and let $Y^\delta_2$ be the multiplication operator defined by (\ref {eq8}).

{In order to compare the approximate solution $(X_2^\delta Y_2^\delta)^n \psi_0$ with the solution $S^t \psi_0$ for any $t=n\delta \le T$, where $T>0$ is any fixed positive real number,  we have to assume that 
the solution $S^t \psi_0$ does not blow up. \ We remark that the approximate solution $(X_2^\delta Y_2^\delta)^n \psi_0$ always exists; however, we have to introduce the 
following technical assumption: we assume that
\bee
\max_{j=0,1,\ldots , n-1} \| (X_2^\delta Y_2^\delta)^{n-j-1} S^{(j+1)\delta } \psi_0 \|_{L^\infty} \le C \label {techass}
\eee
for some positive constant $C$ depending on $\psi_0$ and $T$, but independent of $t$ and $n$. \ We should remark that for each index $j$ the vector $ (X_2^\delta Y_2^\delta)^{n-j-1} S^{(j+1)\delta } \psi_0$ belongs to $L^\infty$ because of Lemma \ref {Lemma3} and Lemma \ref {Lemma4}, the technical assumption concerns the uniformity of the bound with respect to $n$. \ 
Assumption (\ref {techass}) is necessary when we make use of the estimate obtained in Lemma \ref {Lemma5} and its is a rather usual kind of assumption in such a contest (see, e.g., equation (2.4c) in Lemma 2.3 by \cite {Besse} and its application in equation (16) of the same paper, see also the assumptions of Theorem 1 by \cite {S}).}

Here we state our main result. 

\begin {theorem} \label {Teo_uno} {Let $\sigma \ge \frac 12$;} let $T>0$ be any fixed positive real number and let $\psi_0 \in \Ss$ be such that $S^t \psi_0 \in \Ss$ {for any} $t \in [0,T]$. \ Let  $\delta >0$ and $n \in \N$ such that $t= n \delta \le T$. \ Let (\ref {techass}) holds true. \ Then, there exists a positive 
constant $C:=C(\psi_0 , T)$ depending on $\psi_0$ and $T$ such that 
\bee
\left \| \left [ X^\delta_2 Y^\delta_2 \right ]^n \psi_0 - S^{n\delta } \psi_0 \right \|_{L^2} \le C \delta |\nu |\, . \label {teostat}
\eee
\end {theorem}

\begin {remark}
In fact, we expect that such a result may be extended to subquadratic potentials $V(x) \in C^\infty (\R )$ such that $\left \| \frac {\partial^r V(x)}{\partial x^r}
\right \|_{L^\infty} \le C$ as soon as $r\ge 2$. \ {Also the extension to an higher dimension $d\ge 2$ could be considered, too. \ However, in both two cases we have to face some problems: e.g. the proof of Lemma \ref {Lemma3} is based on the 
explicit expression of the propagator $X^t$ of the linear Schr\"odinger operator. \ Some results \cite {Blinder} concerning the generalized Mehler's formula 
could be the basis for such an extension.} \ However, we don't dwell here on the details concerning these two generalizations.
\end {remark}

\section {Proof of Theorem \ref {Teo_uno}} \label {Sez3}
Hereafter, in this Section we simply denote $X_2$ and $Y_2$ respectively by $X$ and $Y$.

\subsection {Preliminary results}

We require some preliminary Lemmas and Remarks.

\begin {lemma} \label {Lemma1}
${\Ss} \subseteq L^p$ for any $p\in [1,+\infty]$. \ In particular
\bee
\| w \|_{L^1 (\R )} \le C \left [ \| x^2 w \|_{L^2(\R )} + 
\| w\|_{L^2 (\R )} \right ] \, ,  \label {eq11}
\eee
where $C= (2^5/3)^{1/8}$.
\end {lemma}

\begin {proof}
The statement ${\Ss} \subseteq L^p$ holds true for $p=+\infty$ by making use of the Gagliardo-Nirenberg inequality:
\be
\| w \|_{L^\infty} \le C \| w \|_{L^2}^{\frac 12} \| w' \|_{L^2}^{\frac 12} \le C \| w \|_{{\Ss}} .
\ee
If we are able to prove that the statement holds true for $p=+1$ too, then the Riesz-Thorin interpolation Theorem prove the statement for any 
$p \in [+1,+\infty ]$. \ In order to prove the statement when $p=+1$ we observe that for any $R>0$
\be
\| w \|_{L^1(\R )} &=& \left [ \int_{-\infty}^{-R} |w(x)| dx + \int_{+R}^{+\infty} |w(x)| dx + \int_{-R}^{+R} |w(x)| dx \right ] \\
&=& \left [ \int_{-\infty}^{-R} \frac {1}{x^2} x^2|w(x)| dx + \int_{+R}^{+\infty}  \frac {1}{x^2} x^2|w(x)| dx + \int_{-R}^{+R} |w(x)| dx \right ] \\
&= &  \langle x^{-2}, x^2 w \rangle_{L^2 (-\infty , - R)} + \langle x^{-2}, x^2 w \rangle_{L^2 (R,\infty )}  + \langle 1, |w| \rangle_{L^2(-R,+R)} \\ 
&\le & \frac {2}{\sqrt {3}} R^{-3/2} \| x^2 w \|_{L^2(\R )} + \sqrt {2R} \| w \|_{L^2 (\R )} < + \infty 
\ee
from the H\"older's inequality. \ Hence, (\ref {eq11}) follows for $R=(2/3)^{1/4}$.
\end {proof}

\begin {remark}
From Lemma \ref {Lemma1} it follows that
\be
\| w \|_{L^1} \le C \| w \|_{\Ss} \, . 
\ee
\end {remark}

The following result holds true

\begin {lemma}\label {Lemma2}
Let $w \in {\Ss}$ then  
\be
\| e^{-iHt} w - w \|_{L^2} \le C| t| \| w \|_{\Ss}
\ee
where $C= \max [1,|\alpha |]$.
\end {lemma}

\begin {proof}
Indeed,  since $w \in {\Ss} \subset {\mathcal D}$ where ${\mathcal D}$ is the self-adjointness domain of $H$, then the evolution 
$v_t(x):=\left [e^{-itH} w\right ] (x) \in {\mathcal D}$ is such that 
\be
\| e^{-itH} w - w \|_{L^2} &=& \| v_t - v_0 \|_{L^2} = \left \| \int_0^t \dot v_{\tau} d\tau \right \|_{L^2} = \left \| \int_0^t i H v_{\tau} d\tau \right \|_{L^2} \\ 
&=& \left \| \int_0^t i H e^{-i\tau H}w d\tau \right \|_{L^2} 
=  \left \| \int_0^t  e^{-i\tau H} H w d\tau \right \|_{L^2} \\
&\le & | t |\,  \left \|  H w \right \|_{L^2} \le | t | \,  \left [ \|  w'' \|_{L^2} + |\alpha | \, \| x^2 w \|_{L^2} \right ]\, , 
\ee
since the two operators $H$ and $e^{-itH}$ commute: $[H,e^{-itH} ]=0$.
\end {proof}

Furthermore, we have that

\begin {lemma} \label {Lemma3}
Let $w \in {\Ss}$, then $X^t w \in {\Ss}$ {for any} $t \in [0,T] $. \ In particular:
\bee
\left \| X^t w  \right \|_{{\Ss}} \le C \left \| w  \right \|_{{\Ss}} \, ,\label {eq12}
\eee
for some positive constant $C>0$ independent of $t$ and $w$.
\end {lemma}

\begin {proof}
Assume, for argument's sake, that $\alpha = + \frac 14 \omega^2$. \ Now, let $a>0$ be fixed and small enough, and let us consider, at first, the case where 
$a \le \left | t- n \frac {\pi}{\omega} \right | \le \frac {\pi }{\omega} -a$, $n \in \Z$. \ Let us recall that 
\be
\left ( X^t w \right ) (x) &:=& \left [ e^{-itH}w \right ] (x)= \int_{\R} K_{HO}(x,y,t) w(y) dy \\ 
&=& 
\sqrt {\frac {\omega}{4\pi i \sin (\omega t)}} \int_{\R} e^{ i\frac {\omega}{4\sin (\omega t)} \left [ (x^2+y^2) \cos (\omega t) - 2 xy \right ] } w(y) dy 
\ee
from the Mehler's formula (\ref {MHO}). \ Hence, for any positive integer $n$
\be
&& x^n\left [ e^{-itH}w \right ] (x) = \sqrt {\frac {\omega}{4\pi i \sin (\omega t)}} \int_{\R} x^n e^{ i\frac {\omega}{4\sin (\omega t)} 
\left [ (x^2+y^2) \cos (\omega t) - 2 xy \right ] } w(y) dy \\
&& \ \ = \frac {1}{\sqrt {2\pi}} 
\left [ \frac { i{2\sin (\omega t)} }{\omega} \right ]^{n-\frac 12} e^{ i\frac {\omega x^2 \cos (\omega t) }{4\sin (\omega t)}  } 
\int_{\R} e^{- i\frac {\omega xy}{2\sin (\omega t)}   } \frac {  \partial^n \left [ e^{ i\frac {\omega y^2 \cos (\omega t) }{4\sin (\omega t)}  } w(y)\right ] }
{\partial y^n}dy 
\ee
In particular, for $n=1$ and $n=2$ it turns out that 
\be
x\left [ e^{-itH}w \right ] (x)&=&  \int_{\R} K_{HO}(x,y,t) \left [ a_1(t) y w (y) + b_1 (t) w'(y)  \right ] dy \\
&=& \left \{ e^{-itH} \left [ a_1(t) y w (y) + b_1(t)  w'(y)  \right ] \right \} (x) \\
x^2\left [ e^{-itH}w \right ] (x)&=&  \int_{\R} K_{HO}(x,y,t) \left [ a_2(t) y^2 w (y) + b_2(t) y w'(y) + c_2(t) w'' (y) \right ] dy \\
&=& \left \{ e^{-itH} \left [ a_2(t) y^2 w (y) + b_2(t) y w'(y) + c_2(t) w'' (y) \right ] \right \} (x)
\ee
for some bounded functions $a_1(t)$, $b_1(t)$, $a_2(t)$, $b_2(t)$ and $c_2(t)$ since $a \le \left | t- n \frac {\pi}{\omega} \right | \le \frac {\pi}{\omega}-a$. \ Then, we can conclude that 
\be
\left \| x\left [ e^{-itH}w \right ]  \right \|_{L^2} &\le &|a_1(t)| \, \| y w \|_{L^2} + |b_1(t)|\, \|  w' \|_{L^2} \le C \| w\|_{\Ss}\\
\left \| x^2\left [ e^{-itH}w \right ]  \right \|_{L^2} &\le &|a_2(t)| \, \| y^2 w \|_{L^2} + |b_2(t)|\, \| y w' \|_{L^2} + |c_2(t) | \| w'' \|_{L^2} \le C \| w\|_{\Ss}
\ee
for some $C$ since 
\be
\| yw \|_{L^2} \le \| w \|_{L^2}^{1/2} \| y^2w \|_{L^2}^{1/2}
\ee
and
\be 
\| y w' \|_{L^2} \le \frac 12  \left [ \| y^2 w \|_{L^2} +   \| w'' \|_{L^2} \right ]\, . 
\ee
Indeed, the last inequality follows by observing that
\be
\left \| y w' \right \|_{L^2}^2 = \langle y w' , y w' \rangle = - \langle 2y w',w \rangle - \langle y^2 w'' , w \rangle 
\ee
and thus
\be
\left \| y w' \right \|_{L^2}^2 \le 2 \| y w' \|_{L^2} \| w \|_{L^2} + \| w'' \|_{L^2} \| y^2 w \|_{L^2} \, .
\ee

Similarly, if one notices that
\be
\frac {\partial }{\partial x} \left [ e^{-itH}w \right ] (x)&=& a_3(t) x \left [ e^{-itH}w \right ] (x) + b_3(t)
 \left [ e^{-itH}xw \right ] (x)
 \ee
for some bounded functions $a_3(t)$ and $b_3(t)$, then the same arguments as above prove that 
\be
\left \| \frac {\partial^2}{\partial x^2} \left [ e^{-itH}w \right ]  \right \|_{L^2} \le C \| w\|_{\Ss}
\ee
for some $C>0$.

Now, one can check that (\ref {eq12}) holds true for any $t$; 
indeed if $t$ is such that $|t| < a$ then we observe that
\be
e^{-itH} w = e^{iaH} e^{-i(t+a)H} w
\ee
from which, since $a \le |t+a| \le \frac {\pi}{\omega} - a$ if $0<t<a$ and $a$ is small enough,  
\be
\left \| e^{-itH} w \right \|_{\Ss} = \left \| e^{iaH} e^{-i(t+a)H} w \right \|_{\Ss} \le C \left \| e^{-i(t+a)H} w \right \|_{\Ss} \le 
C^2 \left \| w \right \|_{\Ss}\, . 
\ee
The case $\left | t- n\frac {\pi}{\omega} \right | <a$, $n \in \Z$, follows in the same way, too.

Eventually, the case $\alpha = - \frac 14 \omega^2<0$ is similarly treated by making use of (\ref {MIO}). 
\end {proof}

Concerning the evolution operator 
\be
\left ( Y^{t} w \right ) (x):=e^{-i\nu |w(x)|^{2\sigma} t} w(x) \, ,
\ee
we recall that
\be
\| Y^{t} w \|_{L^p} = \| w \|_{L^p} \, , \ \forall p \in [1,+\infty ]\, . 
\ee
Furthermore:

\begin {lemma} \label {Lemma4}
Let $w \in {\Ss}$, then $Y^{t} w \in {\Ss}$ for any $t$; in particular
\be
\| Y^{t} w \|_{\Ss} \le  \left [ 1 + C |\nu t | \| w\|_{L^\infty}^{2\sigma} \right ]^2 \, \| w \|_{\Ss} \, .
\ee
for some positive constant $C>0$ independent of $t$ and $w$.
\end {lemma}

\begin {proof}
A straightforward calculation proves that
\be
\| x^2 Y^t w \|_{L^2} = \| x^2 w \|_{L^2}
\ee
and that
\bee
\left \| \frac {\partial^2  Y^t w}{\partial x^2} \right \|_{L^2} \le  \left [ 1 +C | \nu t |\, \| w\|_{L^{\infty}}^{2\sigma} 
 \right ]^2 \| w \|_{H^2} \, .\label {eq13} 
\eee
Indeed,
\be
\left \| \frac {\partial^2  Y^t w}{\partial x^2} \right \|_{L^2} &\le & \| w'' \|_{L^2} +  C \left [ |t \nu |\,  \| w^{2\sigma} w'' \|_{L^2} +\nu^2 t^2 \| w^{4\sigma -1} 
\left ( w' \right )^2 \|_{L^2} +|\nu t| \| w^{2\sigma -1} \left ( w' \right )^2 \|_{L^2} \right ] \\
&\le & \| w'' \|_{L^2} +  C \left [ |\nu t| \| w\|_{L^{\infty}}^{2\sigma} \| w'' \|_{L^2} + \nu^2 t^2 \| w\|_{L^{\infty}}^{4\sigma-1} \|\left ( w' \right )^2 \|_{L^2} +
|\nu  t|\,  \| w\|_{L^{\infty}}^{2\sigma-1} \| \left ( w' \right )^2 \|_{L^2} \right ]
\ee
{since $\sigma \ge 1/2$.} \ Concerning the term $\|\left ( w' \right )^2 \|_{L^2}$ we have that
\be
\|\left ( w' \right )^2 \|_{L^2}^2 &=&\left | \int_{\R} \left ({\bar w'}\right )^2 \left ({w'}\right )^2 dx \right | = 
\left | - \int_{\R} w \left [2 w' \bar w' \bar w'' + w''  
\left ({\bar w'}\right )^2 \right ] dx \right | \\
& \le & 3 \| w \|_{L^\infty} \int_{\R} |w'' | \, \left |{ w'}\right |^2 dx \le C \| w \|_{L^\infty} \| w'' \|_{L^2} \| (w')^2 \|_{L^2}\, ;  
\ee
hence
\bee
\|\left (  w' \right )^2 \|_{L^2} &\le & C \| w \|_{L^\infty} \| w'' \|_{L^2} \, . \label {eq14}
\eee
Thus, we conclude that
\be
\left \| \frac {\partial^2  Y^t w}{\partial x^2} \right \|_{L^2} \le 
 \| w'' \|_{L^2} + C \left [ 2|\nu t|\, \| w\|_{L^{\infty}}^{2\sigma} 
+ \nu^2 t^2 \| w\|_{L^{\infty}}^{4\sigma} \right ] 
\| w'' \|_{L^2} 
\ee
from which (\ref {eq13}) follows.
\end {proof}

The evolution operator $Y^t$ satisfies to the Lipschitz condition, too (see Lemmas 2 and 3 \cite {S}).

\begin {lemma} \label {Lemma5}
Let $w_1,w_2 \in L^2 \cap L^\infty$ and let
\be
M:= \max \left [ \| w_1 \|_{L^\infty} , \| w_2 \|_{L^\infty} \right ]\, . 
\ee
Then, 
\be
\| Y^t w_1 - Y^t w_2 \|_{L^2} \le \left [ 1 + 2\sigma |\nu t|  M^{2\sigma -1} \right ] \| w_1 - w_2 \|_{L^2}\, . 
\ee
\end {lemma}

\begin {remark} \label {R_Pluto} Since the linear operator $X^t := e^{-itH} $  is unitary from $L^2$ to $L^2$ then the Lipschitz condition  
\be
\| X^t Y^t w_1 - X^t Y^t w_2 \|_{L^2} \le \left [ 1 + 2\sigma |\nu t|  M^{2\sigma -1} \right ] \| w_1 - w_2 \|_{L^2} 
\ee
holds true.
\end {remark}

Finally.

\begin {lemma} \label {Lemma6}
Let $F(w ):=|w |^{2\sigma } w $, $w \in {\Ss}$, then $F( w) \in {\Ss}$; in particular
\be
\| F(w) \|_{\Ss} \le C  \| w\|_{L^{\infty}}^{2\sigma} \| w \|_{\Ss}\, ,
\ee
for some positive constant $C>0$ independent of $w$.
\end {lemma}

\begin {proof}
At first we consider
\be
\| x^2 F(w) \|_{L^2} = \| x^2 |w|^{2\sigma} w \|_{L^2} \le \| w \|_{L^\infty}^{2\sigma } \| x^2 w \|_{L^2} \le \| w \|_{L^\infty}^{2\sigma } \| w \|_{\Ss}
\ee
and then, similarly, 
\be
\left \| \frac {\partial^2 F(w)}{\partial x^2} \right \|_{L^2} \le C \left [ \| w \|_{L^\infty}^{2\sigma } \, \| w'' \|_{L^2} +  
\| w \|_{L^\infty}^{2\sigma -1} \| (w')^2 \|_{L^2} \right ] \le C  \| w \|_{L^\infty}^{2\sigma } \, \| w'' \|_{L^2}
\ee
since (\ref {eq14}).
\end {proof}

\begin {remark}\label {RefS}
Finally, we recall here some previous technical results. \ In particular in Lemma 4 by \cite {S} we proved that 
\bee
\| F (w_1 ) - F(w_2 ) \|_{L^2} \le (2\sigma +1)M^{2\sigma} \| w_1 - w_2 \|_{L^2} \, , 
\label {eq15}
\eee
where $M= \max \left [ \| w_1 \|_{L^\infty } , \| w_2 \|_{L^\infty} \right ]$. 
\end {remark}

\subsection {Estimate of the remainder term}

Now, let $S^t$ be the evolution operator associated to the Cauchy problem (\ref {eq10}); it satisfies to the mild equation
\be
\psi_t &=& S^t \psi_0 = 
X^t \psi_0 - i \nu \int_0^t X^{t-s} |\psi_s |^{2\sigma} \psi_s ds \\
&=& X^t \psi_0 - i \nu \int_0^t X^{t-s} F [S^s (\psi_0) ] ds \, . 
\ee
Now, we are going to compare $S^t \psi_0$ with $X^t Y^t \psi_0$ where $Y^t$ satisfies to the mild equation
\be
Y^t \psi_0 = \psi_0 - i \nu \int_0^t F \left [ Y^s (\psi_0 ) \right ] ds \, . 
\ee

Then, we prove that

\begin {theorem} \label {Teo_due}
Let $w\in {\Ss}$ and let $T>0$ be fixed, then 
\be
\| S^t w - X^t Y^t w \|_{L^2} \le  |\nu | C_2 t^2 e^{C_1 t}  \, ,\ \forall t \in [0,T]\, ,
\ee
where $C_1$ and $C_2$ are two positive constants given by:
\bee
C_1 := C_1(w,t) = |\nu| (2\sigma +1) \max_{s \in [0,t]} \left \{ \max \left [ \| S^s w \|_{L^\infty} ,  \| X^s Y^s w \|_{L^\infty} \right ] \right \}^{2\sigma +1} 
\, ,  \label {eq16}
\eee
and
\bee
C_2 :=C_2 (w) = C  \| w \|_{\Ss}^{2\sigma +1} 
\max \left [ 1, T^2 \nu^2 \| w\|_{\Ss}^{4\sigma} \right ]^{2\sigma +1}\, , \label {eq17}
\eee
where $C>0$ is a positive constant independent of $w$, $t$, $\nu$ and $T$.
\end {theorem}

\begin {remark} \label {Remark5}
Indeed, if we assume that $S^s (w)$ {does not blow up} for $s \in [0,T]$ then $\| S^s (w ) \|_{L^\infty}$ is uniformly bounded on time; furthermore, from Lemmas 
\ref {Lemma3} and \ref {Lemma4}, we already known that
\be
\| X^s Y^s w \|_{L^\infty} &\le & \| X^s Y^s w \|_{\Ss} \le C \| Y^s w \|_{\Ss} \\
&\le & C  \left [ 1 +C |s \nu | \| w\|_{L^\infty}^{2\sigma} \right ]^2 \| w\|_{\Ss} \, . 
\ee
Hence $C_1 (w) <+\infty$.
\end {remark}

\begin {proof}
Let $w\in {\Ss}$, then we have that
\bee
&& S^t w - X^tY^t w = - i \nu \left [ \int_0^t X^{t-s} F\left [ S^s (w) \right ] ds - \int_0^t X^t F \left [ Y^s (w ) \right ] ds \right ] \nonumber \\
&&\ \ = - i \nu \int_0^t X^{t-s} \left \{ F \left [ S^s (w ) \right ] - F \left [ X^s Y^s (w ) \right ] \right \} ds + {\mathcal R} (t,w ) \label {eq18}
\eee
where
\be
{\mathcal R} (t,w) = - i \nu \int_0^t X^{t-s} {\mathcal R}_I (s,w) ds 
\ee
and 
\be 
{\mathcal R}_I (s,w) = F \left [ X^s Y^s w \right ] - X^s F \left [ Y^s (w) \right ]\, . 
\ee

\begin {lemma} \label {Lemma7} Let 
\be
M_s:= \max \left [ \| X^s Y^s w \|_{L^\infty} , \|  Y^s w \|_{L^\infty} \right ] \, . 
\ee
Then
\be
\| {\mathcal R}_I (s,w) \|_{L^2} \le C |s| M_s^{2 \sigma} \max \left [ 1, s^2 \nu^2 M_s^{4\sigma} \right ] \| w \|_{\Ss}
\ee
for some positive constant $C>0$ independent of $s$, $\nu$ and $w$.
\end {lemma}

\begin {proof}
Indeed,
\be
\| {\mathcal R}_I (s,w) \|_{L^2} &=& 
\| F \left [ X^s Y^s w \right ] - X^s F \left [ Y^s (w) \right ] \|_{L^2} \\ 
&\le & \| F \left [ X^s Y^s w \right ] - F \left [ Y^s w \right ] \|_{L^2} +\|  X^s F \left [ Y^s (w) \right ]  - F \left [ Y^s w \right ] \|_{L^2} 
\ee
where from (\ref {eq15}) and from Lemma \ref {Lemma2} it follows that 
\be
\| F \left [ X^s Y^s w \right ] - F \left [ Y^s w \right ] \|_{L^2} &\le &(2\sigma +1) M_s^{2\sigma} 
\|  X^s Y^s w - Y^s w \|_{L^2} \\
&\le & (2\sigma +1) M_s^{2\sigma} C |s| \| Y^s w \|_{\Ss}\, .
\ee
Concerning the other term we apply, at first, Lemma \ref {Lemma2} and then Lemma \ref {Lemma6} obtaining that 
\be
\|  X^s F \left [ Y^s (w) \right ]  - F \left [ Y^s w \right ] \|_{L^2} 
&\le & C |s| \| F \left [ Y^s w \right ] \|_{\Ss} \le C| s| \| Y^s w \|_{L^\infty}^{2 \sigma} \| Y^s w \|_{\Ss} \\ 
&\le &  C| s|M_s^{2 \sigma} \| Y^s w \|_{\Ss}
\ee
Hence, we have proved that
\be
\| {\mathcal R}_I \|_{L^2} \le  C| s | M_s^{2 \sigma} \| Y^s w \|_{\Ss} \, . 
\ee
From this result and from Lemma \ref {Lemma4} the proof follows.
\end {proof}

\begin {remark}
From Remark \ref {Remark5} it follows that
\be
M_s \le 
\max \left [ 1 + C |s \nu | \| w\|_{L^\infty}^{2\sigma} \right ]^2
 \| w \|_{\Ss}\, . 
\ee
Thus
\be
\| {\mathcal R}_I (s,w) \|_{L^2} \le C |s| \max \left [ 1, s^2 \nu^2 \| w\|_{\Ss}^{4\sigma} \right ]^{2\sigma +1} \| w \|_{\Ss}^{2\sigma +1}
\ee
for some positive constant $C>0$ independent of $s$ and $w$.
\end {remark}

Then, an estimate of the term ${\mathcal R}$ will follow

\begin {lemma} \label {Lemma8}
Let $w \in {\Ss}$, then
\be
\| {\mathcal R}(t,w) \|_{L^2} \le |\nu | C_2 (w) t^2\, . 
\ee
\end {lemma}

\begin {proof}
Indeed, let $t \ge 0$ for argument's sake; then:
\be
\| {\mathcal R}(t,w) \|_{L^2} & \le & |\nu| \int_0^t \| X^{t-s} {\mathcal R}_I (s,w) \|_{L^2} ds \\ 
& \le & |\nu| \int_0^t \| {\mathcal R}_I (s,w) \|_{L^2} ds \\
& \le & |\nu| \int_0^t C s \max \left [ 1, s^2 \nu^2  \| w\|_{\Ss}^{4\sigma} \right ]^{2\sigma +1}
 \| w \|_{\Ss}^{2\sigma +1}  ds 
\ee
from which the Lemma follows.
\end {proof}

Now, we are ready to estimate the difference  (\ref {eq18}): 
\be
&& \| S^t w - X^t Y^t w \|_{L^2} \le  |\nu | \int_0^t \left \| X^{t-s} \left \{ F \left [ S^s w \right ] - F \left [ X^s Y^s w  \right ] \right \} 
\right \|_{L^2}  ds + \| {\mathcal R} (t,w)\|_{L^2} \\ 
&& \ \le |\nu | \int_0^t \left \|  F \left [ S^s w \right ] - F \left [ X^s Y^s w \right ]  \right \|_{L^2}  ds + \| {\mathcal R} (t,w)\|_{L^2} \\ 
&&\ \ \le |\nu |(2\sigma +1)  \int_0^t \max \left [ \| S^s w \|_{L^\infty} ,  \| X^s Y^s w \|_{L^\infty} \right ]^{2\sigma } 
\left \|   S^s w -  X^s Y^s w \right \|_{L^2}  ds + \| {\mathcal R} (t,w)\|_{L^2} \\ 
&&\ \ \le C_1(w,t ) \int_0^t \left \|   S^s w -  X^s Y^s w \right \|_{L^2}  ds +|\nu | C_2 (w) t^2
\ee
from Remark \ref {RefS}, recalling that $X^t$ is an unitary operator on $L^2$ and where $C_1 (w,t)$ and $C_2 (w)$ are respectively defined by (\ref {eq16}) and (\ref {eq17}). \ That is
\be
y(t) \le C_1 \int_0^t y(s) ds + |\nu | C_2 t^2   \, , \ t \in [0,T]\, , 
\ee
where we set
\be
y(t) := \| S^t w - X^t Y^t w \|_{L^2}\, . 
\ee
Thus, the Gronwall's Lemma implies that
\be
y(t) \le  |\nu | C_2 t^2 e^{C_1 t}  \, ,\ \forall t \in [0,T]\, .
\ee
and Theorem \ref {Teo_due} is so proved.
\end {proof}

Finally, we can conclude the proof of Theorem \ref {Teo_uno}. \ Let us fix $t\le T$, let $\delta >0$ small enough and let $n \in \N$ such that $t=n\delta$, let $Z^\delta = X^\delta Y^\delta$; then, the triangle inequality 
yields to 
\be
\left \| \left (Z^{\delta } \right )^n \psi_0 - S^{n\delta } \psi_0 \right \|_{L^2} 
&=& \left \| \sum_{j=0}^{n-1} \left [ \left (Z^{\delta} \right )^{(n-j-1)} Z^\delta S^{j\delta } \psi_0 - \left (Z^{\delta} \right )^{(n-j-1)} S^{(j+1)\delta } 
\psi_0 \right ] \right \|_{L^2} \\
&\le & \sum_{j=0}^{n-1} \left \| Z^\delta \left (Z^{\delta} \right )^{(n-j-1)}  S^{j\delta } \psi_0 - Z^\delta \left (Z^{\delta} \right )^{(n-j-2)} S^{(j+1)\delta } 
\psi_0 \right \|_{L^2} \, . 
\ee

From this inequality, by making use of Lemma \ref {Lemma5}, Remark \ref {R_Pluto} and Theorem \ref {Teo_due} it follows that
\bee
c_{n-j-1,j}&:=& \left \| Z^\delta \left (Z^{\delta} \right )^{(n-j-1)}  S^{j\delta } \psi_0 - Z^\delta \left (Z^{\delta} \right )^{(n-j-2)} S^{(j+1)\delta } 
\psi_0 \right \|_{L^2} \nonumber \\ 
& \le &  \left [ 1 + 2 \sigma |\nu \delta | C^{2\sigma -1} \right ] c_{n-j-2,j}\nonumber \\
&\le &  \left [ 1 + 2 \sigma |\nu \delta | C^{2\sigma -1} \right ]^{n-j-1} c_{0,j} \label {new}
\eee
for some positive constant $C$ independent of $n$ since (\ref {techass}). \ Therefore, we have proved that 
\be
\left \| \left (Z^{\delta } \right )^n \psi_0 - S^{n\delta } \psi_0 \right \|_{L^2}  
&\le & \sum_{j=0}^{n-1} \left [ 1 + 2 \sigma |\nu \delta | C^{2\sigma -1} \right ]^{n-j-1} 
\left \|  Z^\delta S^{j\delta } \psi_0 -  S^{(j+1)\delta } \psi_0 \right \|_{L^2} \\
&\le & \sum_{j=0}^{n-1} \left [ 1 + 2 \sigma |\nu \delta | C^{2\sigma -1} \right ]^{n-j-1} |\nu |C_{2,j} \delta^2 e^{C_{1,j} \delta }
\ee
where
\be
C_{2,j} &:=& C_2 (S^{j\delta }\psi_0 ) =  C  \|S^{j\delta} \psi_0 \|_{\Ss}^{2\sigma +1} \max \left [ 1 , T^2 \nu^2 
\left \| S^{j\delta } \psi_0 \right \|_{{\Ss}}^{4\sigma } \right ]^{2\sigma +1} \\ 
C_{1,j} &:=& C_1 (S^{j\delta }\psi_0 , \delta) = |\nu |(2\sigma +1) \max_{s\in [0,\delta ]} \left \{ \max \left [ 
\left \|  S^{(j+1)\delta } \psi_0 \right \|_{L^\infty} , 
\left \|  Z^s S^{j\delta } \psi_0 \right \|_{L^\infty}
\right ] \right \} \, . 
\ee
From Lemma \ref {Lemma3} and Lemma \ref {Lemma4} then it follows that
\be
\left \| Z^s S^{j\delta } \psi_0 \right \|_{L^\infty} \le \left \| Z^s S^{j\delta } \psi_0 \right \|_{{\Ss}} \le 
C \max \left [ 1 , \nu^2 \delta^2 \left \| S^{j\delta } \psi_0 \right \|_{L^\infty}^{4\sigma } 
\right ] \left \| S^{j\delta } \psi_0 \right \|_{{\Ss}}
\ee
and that
\be
\left \| S^{(j+1)\delta} \psi_0 \right \|_{L^\infty} \le C \left \| S^{(j+1)\delta} \psi_0 \right \|_{\Ss}\, .
\ee
Since we assume that the solution $S^t \psi_0 \in {\Ss}$ for any $t \le T$ then we can conclude that
\be
C_{1,j} \, , \ C_{2,j} \le C_3 \, , \ \forall j =0,1,\ldots , n-1\, , 
\ee
for some positive constant $C_3:=C_3 (T, \psi_0 )>0$. \ Hence,
\be
\left \| \left (Z^{\delta } \right )^n \psi_0 - S^{n\delta} \psi_0 \right \|_{L^2} 
&\le & C |\nu |\delta^2 \sum_{j=0}^{n-1} \left [ 1 + 2 \sigma |\nu \delta | C^{2\sigma -1} \right ]^{n-j-1} C_3  e^{C_3 \delta }\\ 
&\le & C  \delta |\nu t| , \ t=n \delta < T\, ,  
\ee
from which the proof of Theorem \ref {Teo_uno} follows.

\section {Numerical Experiments} \label {Sez4}

{For any fixed $t$ we numerically compute the approximate solutions 
\be
\psi_{t,j} =[X_j^\delta Y^\delta_j ]^n \psi_0
\ee
for different values of $n$ where $\delta = \frac {t}{n}$; $\psi_j  = \left [ X_j^\delta Y_j^\delta \right ]^n \psi_0 $, $j=1,2$, where $X_1^\delta$ is 
the evolution operator associated to $- \frac {\partial^2}{\partial x^2}$, $X_2^\delta$ is the evolution operator associated to $- \frac {\partial^2}{\partial x^2} + V$, 
$Y_1^\delta$ is the evolution operator associated to the differential equation $i \dot \psi = V \psi + \nu |\psi |^{2\sigma} \psi$ and $Y_2^\delta$ is the evolution operator 
associated to the differential equation $i \dot \psi_t =\nu |\psi_t |^{2\sigma} \psi_t$. \ We consider the harmonic oscillator potential where $V(x)= + \frac 14 \omega^2 x^2$ and the inverted oscillator potential where $V(x)=-\frac 14 \omega^2 x^2$. \ In 
both cases we consider the focusing (where $\nu <0$) and defocusing (where $\nu >0$) nonlinearity.

In this Section, for simplicity's sake, let us drop out the index $t$, i.e. $\psi_t = \psi$, $\psi_{t,j} = \psi_j$, and so on. \ For argument's sake the initial wavefunction is a Gaussian function 
\be
\psi_0 (x) = \frac {1}{\sqrt [4] {2\pi \St^2}} e^{-(x-x_0)^2/4 \St^2 + i v_0 x}
\ee
where
\be
x_0 =-3 \, , \ v_0 = 2 \ \mbox { and } \ \St =0.5 \, .
\ee
We compare in numerical experiments the rate of convergence of the numerical solutions $\psi_j$.\ More precisely, we compare the probability densities
\be
\rho_j (x) = |\psi_j (x)|^2 \, , \ j=1,2 \, , 
\ee
and the expectation value of the position observable
\be
\langle x \rangle_j := \langle \psi_j , x \psi_j \rangle_{L^2} \, , \ j=1,2 \, , 
\ee
for a fixed value of $t$. 

We recall that the evolution operators $Y_1^\delta$ and $Y_2^\delta$ are the multiplication operators (\ref {eq5Bis}) and (\ref {eq8}); $X_1^\delta$ and $X_2^\delta$ are the integral operators (\ref {eq18Bis}), (\ref {eq19Bis}) and (\ref {MIO}). \ Since 
the evolution operators $X_j^\delta$ are integral operators then we numerically compute the integral on a large enough fixed interval $[x_{min},x_{max}]$ by 
dividing it in $m$ intervals with the same amplitude $\frac {x_{max}-x_{min}}{m}$, that is $m$ is the number of points of the mesh. \ Let $\psi_j^{n,m}$ be the numerical solutions given by the vector $\left ( \left [ X_j^\delta Y_j^\delta \right ]^n \psi_0 \right )(x_\ell )$, 
where $x_\ell = x_{min}+ \ell \frac {x_{max}-x_{min}}{m} $ for $\ell =0,1,\ldots , m$. 

If we denote by $\psi_j^{\infty}$ and $\rho_j^{\infty}$ the values of 
$\psi_j^{n,m}$ and $\rho_j^{n,m}$, $j=1,2$, where $n$ and $m$ are the largest values considered in the numerical experiment, then we are going to estimate the quantities
\be
\Delta_j^{n,m} = \max_{\ell =0,1,\ldots ,m} | \rho_j^\infty (x_\ell ) - \rho_j^{n,m} (x_\ell ) | \, , \ j=1,2\, , 
\ee
for different values of $n$ and $m$. \ Furthermore, we consider also the difference
\be
\delta^{n,m}:= 
\max_{\ell =0,1,\ldots ,m} | \rho_1^{n,m}(x_\ell ) - \rho_2^{n,m} (x_\ell )| \, . 
\ee
}

Finally, we compare also the exact expected value of the position observable $\langle x \rangle^t = \langle \psi_t , x \psi_t \rangle_{L^2} $ with the ones $\langle x \rangle^t_j = \langle \psi_{j,t} , x \psi_{j,}t \rangle_{L^2}  $, $j=1,2$,  obtained with the two approximate solutions. \ In fact, Ehrenfest's Theorem for nonlinear Sch\"odinger equations does not generically hold true in the usual form, but when one cosniders the position $x$ and momentum $p$ observables we still have that
\be
\frac {d\langle x \rangle^t}{dt} = \frac {1}{m} \langle p \rangle^t \ \mbox { and } \ \frac {d\langle p \rangle^t}{dt} =- \left \langle \frac {dV}{dx} \right \rangle^t 
\ee
where 
\be
\langle p \rangle^t = -i \left \langle \psi_t , \frac {\partial \psi_t}{\partial x} \right \rangle_{L^2} \ \mbox { and } \  \left \langle \frac {dV}{dx} \right \rangle^t =  \left \langle \psi_t, \frac {dV}{dx} \psi_t \right \rangle_{L^2} \, . 
\ee
In particular, since $m= \frac 12$ and $V(x) = \alpha x^2$ then $\langle x \rangle^t$ is solution to the differential equation
\be
\left \{
\begin {array}{l}
\frac {d^2 \langle x\rangle^t}{dt^2} + \alpha \langle x \rangle^t =0 \\
\langle x \rangle^0 = x_0 \ \mbox { and } \ \left. \frac {d \langle x\rangle^t}{dt} \right |_{t=0} = 2 \langle p \rangle^0 = 2 v_0 
\end {array}
\right.
\ee
Thus
\bee
\langle x \rangle^t = 
\left \{
\begin {array}{ll}
x_0 \cos (\omega t ) + \frac {2v_0}{\omega} \sin (\omega t ) & \mbox { if } \alpha = + \frac 14 \omega^2 \\ 
 & \\
x_0 \cosh (\omega t ) + \frac {2v_0}{\omega} \sinh (\omega t ) & \mbox { if } \alpha = - \frac 14 \omega^2
\end {array}
\right. \, . \label {eqA}
\eee

\subsection {Harmonic oscillator}

In such an experiment let
\be
\omega =1 \, , \ \sigma =1 \ \mbox { and } \ t=10\, .
\ee
We numerically compute the integral operators $X_1^\delta$ and $X_2^\delta$ where the integral domain is restricted to the interval $[x_{min} , x_{max}]$ where
\be
x_{min}=-50 \ \mbox { and } \ x_{max} =+ 50 .
\ee
The indexes $n$ and $m$ respectively run from $60$ to $240$ and from $2000$ to $8000$; we denote by $\psi_j^{\infty} = \psi_j^{300,10000}$ the corresponding 
numerical solution obtained when $n=300$, and thus $\delta = \frac {1}{30}$, and $m=10000$.

\subsubsection {Defocusing nonlinearity} We fix
\be
\nu =+10\, .
\ee
{The numerical experiment shows that the following upper bound of the absolute value of the difference between the two probability densities $\rho_1^\infty$ 
and $\rho_2^\infty$ holds true} 
\bee
\max_{x} | \rho_1^\infty - \rho_2^\infty |= 0.025 \, , \label {upper1}
\eee
and in Figure \ref {Fig1} - left hand side - we plot the graph of the function $\rho_2^\infty$. \ In Table \ref {tabella1} we collect the difference $\Delta_j^{n,m}$ 
between $\rho_j^{n,m}$ and $\rho_j^{\infty}$, the ratio $\Delta_2^{n,m}/\Delta_1^{n,m}$, the difference $\delta_j^{n,m}$ between $\rho_1^{n,m}$ and $\rho_2^{n,m}$ and, 
finally, the expectation values $\langle x \rangle_1^t$ and $\langle x \rangle_2^t$ for different values of $n$ and $m$ and for $t=3$. \ It turns out that the 
values obtained in correspondence of the approximation $\psi_2^t$ become rapidly stable even for $n$ and $m$ not particularly large; in particular the 
expectation value $\langle x \rangle_2^{10}$ is practically constant, while the expectation value $\langle x \rangle_1^{10}$ slowly converges to its final value.
\begin{table}
\begin{center}
\begin{tabular}{|c|c||c|c|c||c||c|c|} 
\hline
 \multicolumn{8}{|c|}{$\nu =+10$} 
\\ \hline
$n $  &  $m$   & $\Delta_1^{n,m}$  &   $\Delta_2^{n,m}$   &  $ {\Delta_2^{n,m}}/{\Delta_1^{n,m}}$ & $\delta^{n,m}$ & $\langle x \rangle_1^{10}$ & $\langle x \rangle_2^{10}$
\\ \hline \hline 
$60$      &  $2000$   & $0.31$   & $0.18$   & $0.58$ & $0.16$ & $0.14$ & $0.34$ \\ \hline
$90$      &  $3000$   & $0.12$   & $0.12$   & $0.94$ & $0.08$ & $0.22$ & $ 0.34$\\ \hline
$120$     &  $4000$   & $0.077$  & $0.046$  & $0.60$ & $0.070$ & $0.26$ & $0.34$ \\ \hline
$150$     &  $5000$   & $0.048$  & $0.027$  & $0.56$ & $0.055$ & $0.28$ & $0.34$ \\ \hline
$180$     &  $6000$   & $0.032$  & $0.017$  & $0.53$ & $0.044$ & $0.29$ & $0.34$ \\ \hline
$210$     &  $7000$   & $0.020$  & $0.010$  & $0.51$ & $0.037$ & $0.30$ & $0.34$ \\ \hline
$240$     &  $8000$   & $0.011$  & $0.0060$ & $0.53$ & $0.032$ & $0.30$ & $0.34$ \\ \hline
$270$     &  $9000$   & $0.0049$ & $0.0027$ & $0.55$ & $0.028$ & $0.31$ & $0.34$ \\ \hline
\end{tabular}
\caption{Table of values corresponding to the case of defocusing nonlinearity $\nu=+10$ with harmonic oscillator potential $V(x)=+\frac 14 \omega^2 x^2$. \ The exact expectation value is $\langle x \rangle^{10} = 0.3411$ from (\ref {eqA}).}
\label{tabella1}
\end{center}
\end {table}
\begin{figure}
\begin{center}
\includegraphics[height=5cm,width=5cm]{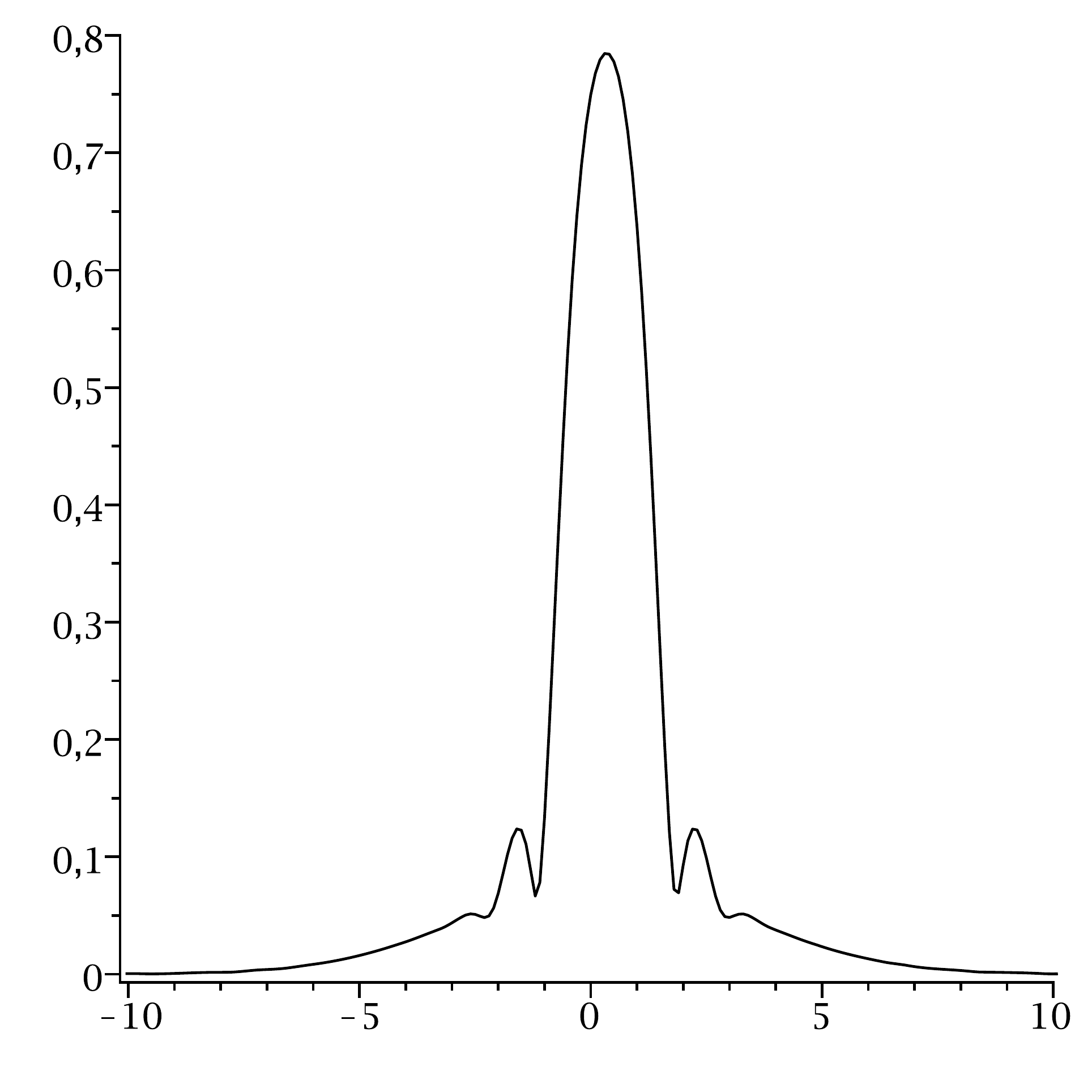}
\includegraphics[height=5cm,width=5cm]{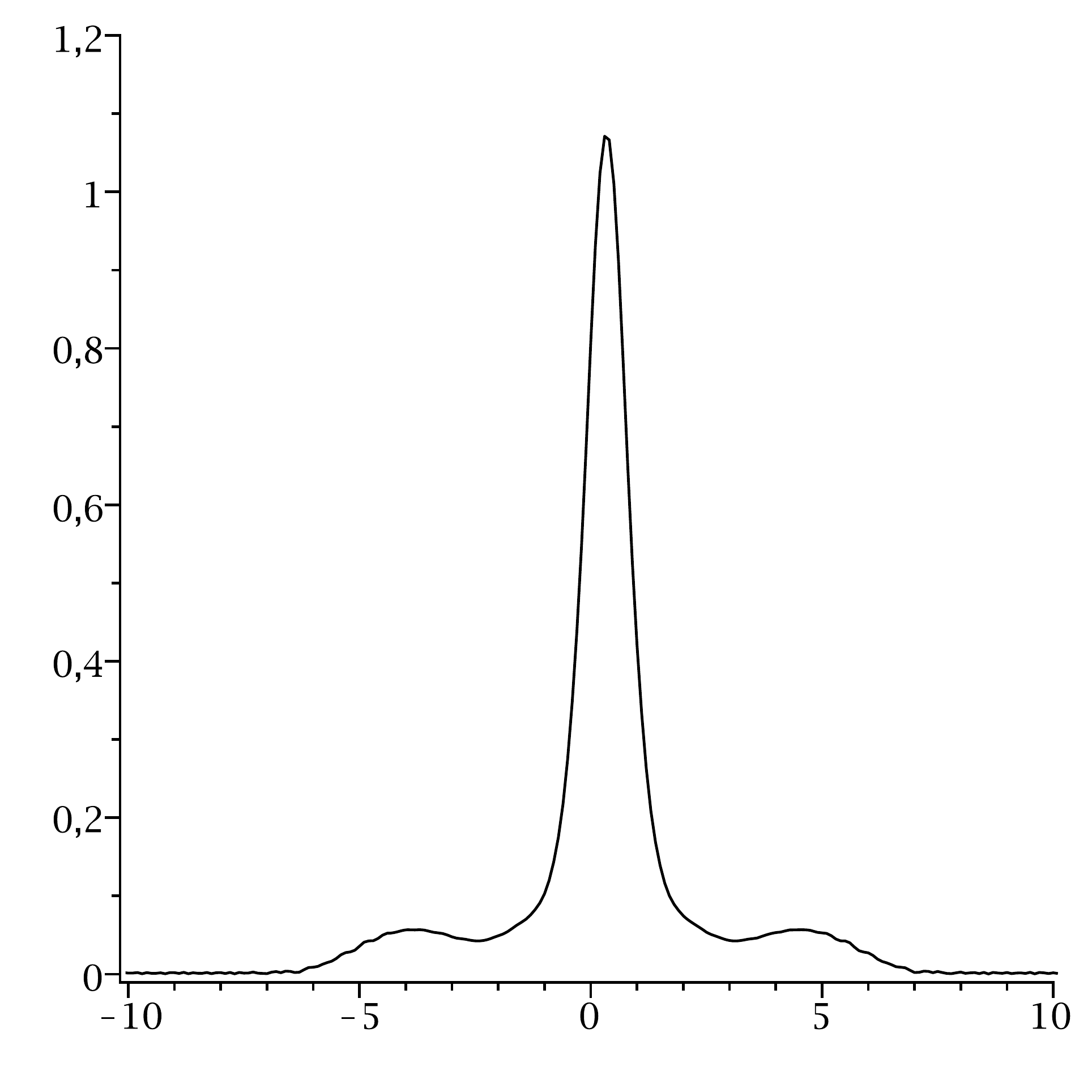}
\caption{Harmonic oscillator. \ We plot the graph of the probability density $\rho_2^\infty $ at $t=10$ of the numerical solution 
$\psi_2^\infty$ in the defocusing case for $\nu =+10$ (left hand side picture) and in the focusing case for $\nu =-10$ (right hand side picture).}
\label {Fig1}
\end{center}
\end{figure}

\subsubsection {Focusing nonlinearity} We fix
\be
\nu =-10\, .
\ee
{In this case the numerical experiment shows that the upper bound concerning the difference between the two probability densities $\rho_1^\infty$ 
and $\rho_2^\infty$ is of the same order of (\ref {upper1}) since}
\be
\max_{x} | \rho_1^\infty - \rho_2^\infty |= 0.040 \, , 
\ee
and in Figure \ref {Fig1} - right hand side - we plot the graph of the function $\rho_2^\infty$. \ In Table \ref {tabella2} we collect the difference 
$\Delta_j^{n,m}$ between $\rho_j^{n,m}$ and $\rho_j^{\infty}$, the ratio $\Delta_2^{n,m}/\Delta_1^{n,m}$, the difference $\delta_j^{n,m}$ 
between $\rho_1^{n,m}$ and $\rho_2^{n,m}$ and, finally, the expectation values $\langle x \rangle_1^t$ and $\langle x \rangle_2^t$ for different 
values of $n$ and $m$ and for $t=3$. \ It turns out that the values for the expectation values coincide with the ones obtained in defocusing case; 
even in this case the approximation $\psi_2^t$ become rapidly stable even for $n$ and $m$ not particularly large and we can observe the same 
behaviour of $\langle x \rangle_1^{10}$ and $\langle x \rangle_2^{10}$ already observed in the defocusing case (in fact, the expectation 
values are exactly the same of the previous experiment).
\begin{table}
\begin{center}
\begin{tabular}{|c|c||c|c|c||c||c|c|} 
\hline
 \multicolumn{8}{|c|}{$\nu =-10$} 
\\ \hline
$n $  &  $m$   & $\Delta_1^{n,m}$  &   $\Delta_2^{n,m}$   &  $ {\Delta_2^{n,m}}/{\Delta_1^{n,m}}$ & $\delta^{n,m}$ & $\langle x \rangle_1^{10}$ & $\langle x \rangle_2^{10}$
\\ \hline \hline 
$60$      &  $2000$   & $0.43$   & $0.29$   & $0.68$ & $0.19$ & $0.14$ & $0.34$ \\ \hline
$90$      &  $3000$   & $0.34$   & $0.34$   & $1.00$ & $0.12$ & $0.22$ & $0.34$ \\ \hline
$120$     &  $4000$   & $0.17$   & $0.17$   & $1.00$ & $0.099$ & $0.26$ & $0.34$ \\ \hline
$150$     &  $5000$   & $0.076$   & $0.054$   & $0.71$ & $0.082$ & $0.28$ & $0.34$ \\ \hline
$180$     &  $6000$   & $0.086$  & $0.083$  & $0.96$ & $0.075$ & $0.29$ & $0.34$ \\ \hline
$210$     &  $7000$   & $0.081$  & $0.077$  & $0.95$ & $0.064$ & $0.30$ & $0.34$ \\ \hline
$240$     &  $8000$   & $0.038$  & $0.036$  & $0.94$ & $0.057$ & $0.30$ & $0.34$ \\ \hline
$270$     &  $9000$   & $0.023$  & $0.022$  & $0.94$ & $0.045$ & $0.31$ & $0.34$ \\ \hline
\end{tabular}
\caption{Table of values corresponding to the case of focusing nonlinearity $\nu=-10$ with harmonic oscillator potential $V(x)=+\frac 14 \omega^2 x^2$. \ The exact expectation value is $\langle x \rangle^{10} = 0.3411$ from (\ref {eqA}).}
\label{tabella2}
\end{center}
\end {table}

\subsection {Inverted oscillator} In such an experiment let
\be
\omega =1 \, , \ \sigma =1 \ \mbox { and } \ t=3\, .
\ee
We numerically compute the integral operators $X_1^\delta$ and $X_2^\delta$ where the integral domain is restricted to the interval $[x_{min} , x_{max}]$ where
\be
x_{min}=-200 \ \mbox { and } \ x_{max} =+ 200 .
\ee
The indexes $n$ and $m$ respectively run from $30$ to $135$ and from $10000$ to $45000$; thus we denote by $\psi_j^{\infty} = \psi_j^{150,50000}$ the corresponding 
numerical solution obtained when $n=150$, and thus $\delta = \frac {1}{50}$, and $m=50000$.

\subsubsection {Defocusing nonlinearity} We fix
\be
\nu =+10\, .
\ee
{In this case the two probability densities $\rho_1^\infty$ and $\rho_2^\infty$ practically coincides since the numerical experiment shows that the upper bound 
of their difference in much smaller than (\ref {upper1}); indeed, it takes the value}
\be
\max_{x} | \rho_1^\infty - \rho_2^\infty |= 0.00046 \, , 
\ee
and in Figure \ref {Fig2} - left hand side - we plot the graph of the function $\rho_2^\infty$. \ In Table \ref {tabella5} we collect the difference 
$\Delta_j^{n,m}$ between $\rho_j^{n,m}$ and $\rho_j^{\infty}$, the ratio $\Delta_2^{n,m}/\Delta_1^{n,m}$, the difference $\delta_j^{n,m}$ between $\rho_1^{n,m}$ 
and $\rho_2^{n,m}$ and, finally, the expectation values $\langle x \rangle_1^t$ and $\langle x \rangle_2^t$ for different values of $n$ and $m$ and for $t=10$. \ It 
turns out that, as well as in the previous experiments, the values obtained in correspondence of the approximation $\psi_2^t$ become rapidly stable even for $n$ and $m$ 
not particularly large.
\begin{table}
\begin{center}
\begin{tabular}{|c|c||c|c|c||c||c|c|} 
\hline
 \multicolumn{8}{|c|}{$\nu =+10$} 
\\ \hline
$n $  &  $m$   & $\Delta_1^{n,m}$  &   $\Delta_2^{n,m}$   &  $ {\Delta_2^{n,m}}/{\Delta_1^{n,m}}$ & $\delta^{n,m}$ & $\langle x \rangle_1^{3}$ & $\langle x \rangle_2^{3}$ \\ \hline \hline 
$30$      &  $10000$   & $0.017$   & $0.0050$   & $0.29$ & $0.017$ & $8.10$ & $9.87$ \\ \hline
$45$      &  $15000$   & $0.0024$   & $0.0017$   & $0.70$ & $0.0015$ & $8.84$ & $9.87$ \\ \hline
$60$     &  $20000$   & $0.0015$   & $0.0011$   & $0.69$ & $0.0011$ & $9.10$ & $9.87$ \\ \hline
$75$     &  $25000$   & $0.0010$   & $0.00069$   & $0.68$ & $0.0009$ & $ 9.26$ & $9.87$ \\ \hline
$90$     &  $30000$   & $0.00067$  & $0.00045$  & $0.67$ & $0.00075$ & $9.36$ & $9.87$ \\ \hline
$105$     &  $35000$   & $0.00043$  & $0.00029$  & $0.67$ & $0.00065$ & $9.43$ & $9.87$ \\ \hline
$120$     &  $40000$   & $0.00025$  & $0.00017$  & $0.66$ & $0.00057$ & $9.49$ & $9.87$ \\ \hline
$135$     &  $45000$   & $0.00011$  & $0.000073$  & $0.66$ & $0.00051$ & $9.53$ & $9.87$ \\ \hline
\end{tabular}
\caption{Table of values corresponding to the case of defocusing nonlinearity $\nu=+10$ with inverted oscillator potential $V(x)=-\frac 14 \omega^2 x^2$. \ The exact expectation value is $\langle x \rangle^{3} = 9.8685$ from (\ref {eqA}).}
\label{tabella5}
\end{center}
\end {table}
\begin{figure}
\begin{center}
\includegraphics[height=5cm,width=5cm]{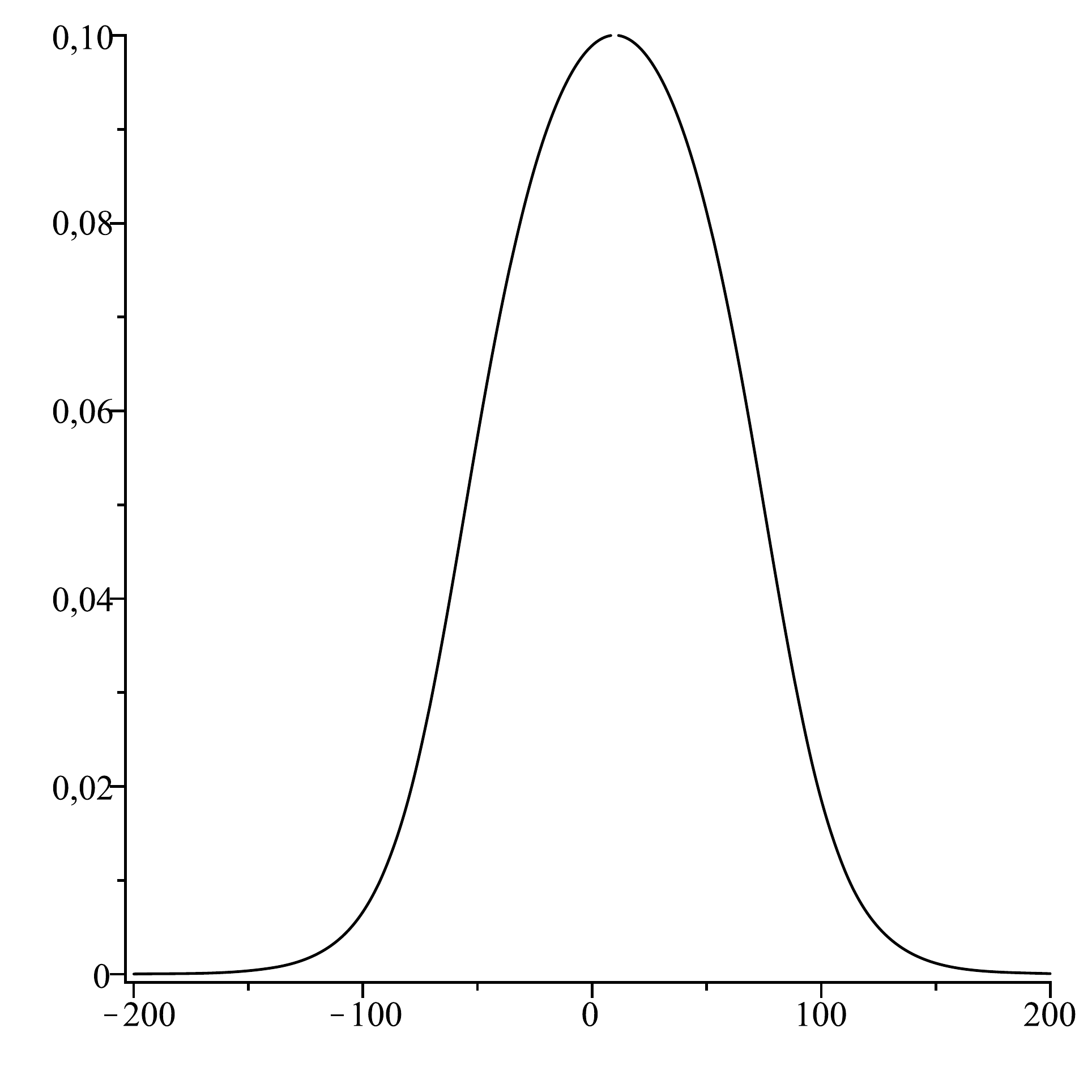}
\includegraphics[height=5cm,width=5cm]{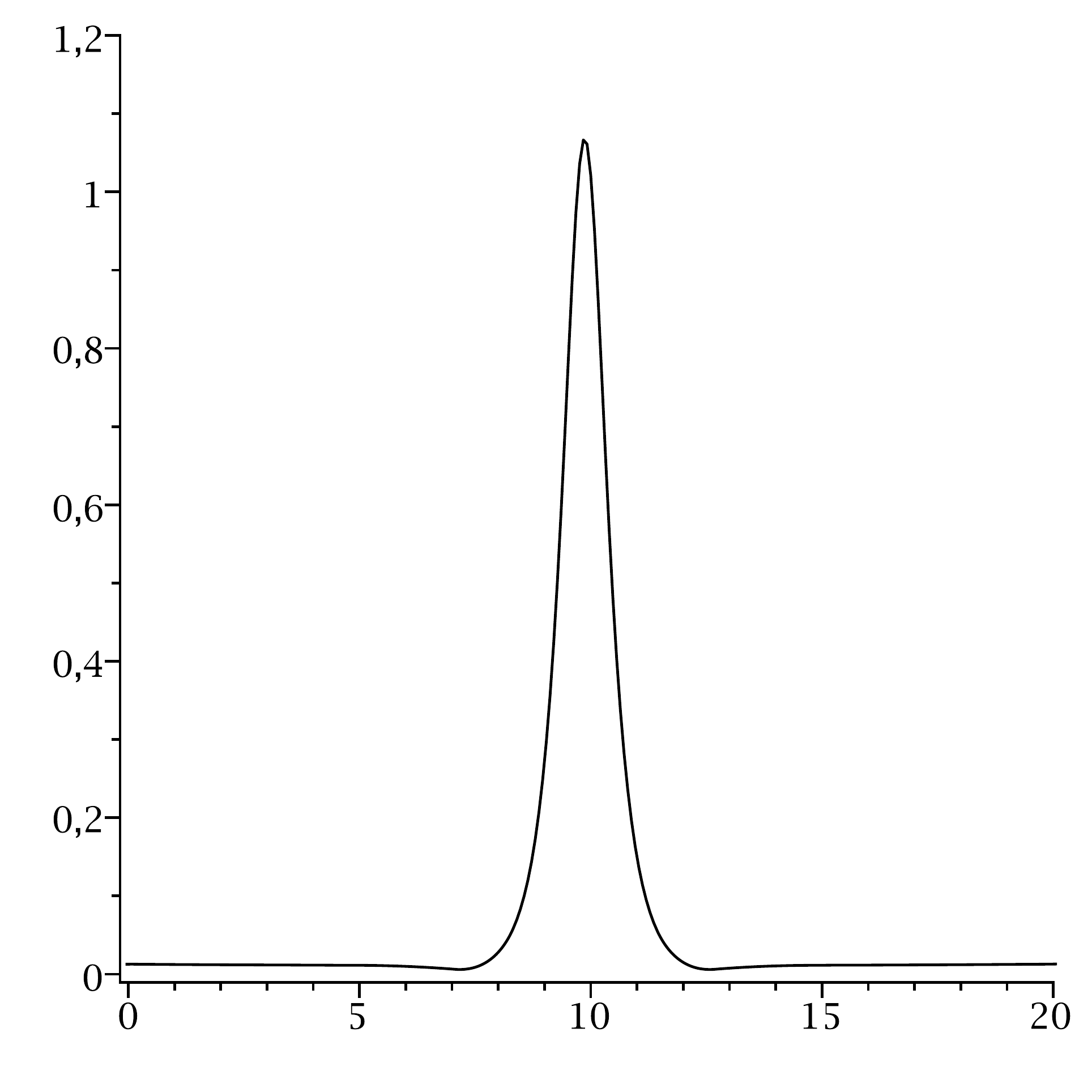}
\caption{Inverted oscillator. \ We plot the graph of the probability density $\rho_2^\infty $ at $t=3$ of the numerical solution $\psi_2^\infty$ in the defocusing 
case for $\nu =+10$ (left hand side picture) and in the focusing case for $\nu =-10$ (right hand side picture).}
\label {Fig2}
\end{center}
\end{figure}

\subsubsection {Focusing nonlinearity} We fix
\be
\nu =-10\, .
\ee
In this experiment {we have that a significant} difference between $\rho_1^\infty$ and $\rho_2^\infty$ occurs because of a shift, as shown in Figure \ref {Fig6},  since 
\be
\max_{x} | \rho_1^\infty - \rho_2^\infty |= 0.37 \, , 
\ee
that slowly decreases for increasing values of $n$ and $m$. \ {Such a shift is due to the fact that in the usual spectral splitting approximation $(X_1^\delta Y_1^\delta)^n \psi_0$ the linear part is approximated as well as the nonlinear one; 
in contrast, in the proposed here spectral splitting approximation $(X_2^\delta Y_2^\delta)^n \psi_0$ the linear part is exactly solved and the approximation only concerns the nonlinear one. \ In order to reduce such 
a shift one has to significantly increases the number $m$ of the points of the mesh and the number $n$ of iterations in the usual spectral splitting approximation.}

In Table \ref {tabella6} we collect the difference $\Delta_j^{n,m}$ between $\rho_j^{n,m}$ and $\rho_j^{\infty}$, the ratio $\Delta_2^{n,m}/\Delta_1^{n,m}$, the 
difference $\delta_j^{n,m}$ between $\rho_1^{n,m}$ and $\rho_2^{n,m}$ and, finally, the expectation values $\langle x \rangle_1^t$ and $\langle x \rangle_2^t$ for 
different values of $n$ and $m$ and for $t=10$. \ Concerning the velocity of convergence of the approximate solutions we can draw the same kind of conclusions of 
the previous numerical experiments.
\begin{table}
\begin{center}
\begin{tabular}{|c|c||c|c|c||c||c|c|} 
\hline
 \multicolumn{8}{|c|}{$\nu =-10$} 
\\ \hline
$n $  &  $m$   & $\Delta_1^{n,m}$  &   $\Delta_2^{n,m}$   &  $ {\Delta_2^{n,m}}/{\Delta_1^{n,m}}$ & $\delta^{n,m}$ 
& $\langle x \rangle_1^{3}$ & $\langle x \rangle_2^{3}$ \\ \hline \hline 
$30$      &  $10000$   & $0.93$   & $0.12$   & $0.13$ & $0.89$ & $8.23$ & $9.86$ \\ \hline
$45$      &  $15000$   & $0.72$   & $0.063$   & $0.087$ & $0.82$ & $8.84$ & $9.87$ \\ \hline
$60$     &  $20000$   & $0.52$   & $0.037$   & $0.071$ & $0.73$ & $9.10$ & $9.87$  \\ \hline
$75$     &  $25000$   & $0.37$   & $0.023$   & $0.061$ & $0.64$ & $9.25$ & $9.87$  \\ \hline
$90$     &  $30000$   & $0.26$  & $0.014$  & $0.055$ & $0.56$ & $9.36$ & $9.87$  \\ \hline
$105$     &  $35000$   & $0.17$  & $0.0086$  & $0.051$ & $0.50$ & $9.43$ & $9.87$  \\ \hline
$120$     &  $40000$   & $0.099$  & $0.0048$  & $0.048$ & $0.45$ & $9.49$ & $9.87$  \\ \hline
$135$     &  $45000$   & $0.044$  & $0.0020$  & $0.046$ & $0.40$ & $9.53$ & $9.87$  \\ 
\hline
\end{tabular}
\caption{Table of values corresponding to the case of focusing nonlinearity $\nu=-10$ with inverted oscillator potential $V(x)=-\frac 14 \omega^2 x^2$. \ The exact expectation value is $\langle x \rangle^{3} = 9.8685$ from (\ref {eqA}).}
\label{tabella6}
\end{center}
\end {table}
\begin{figure}
\begin{center}
\includegraphics[height=6cm,width=8cm]{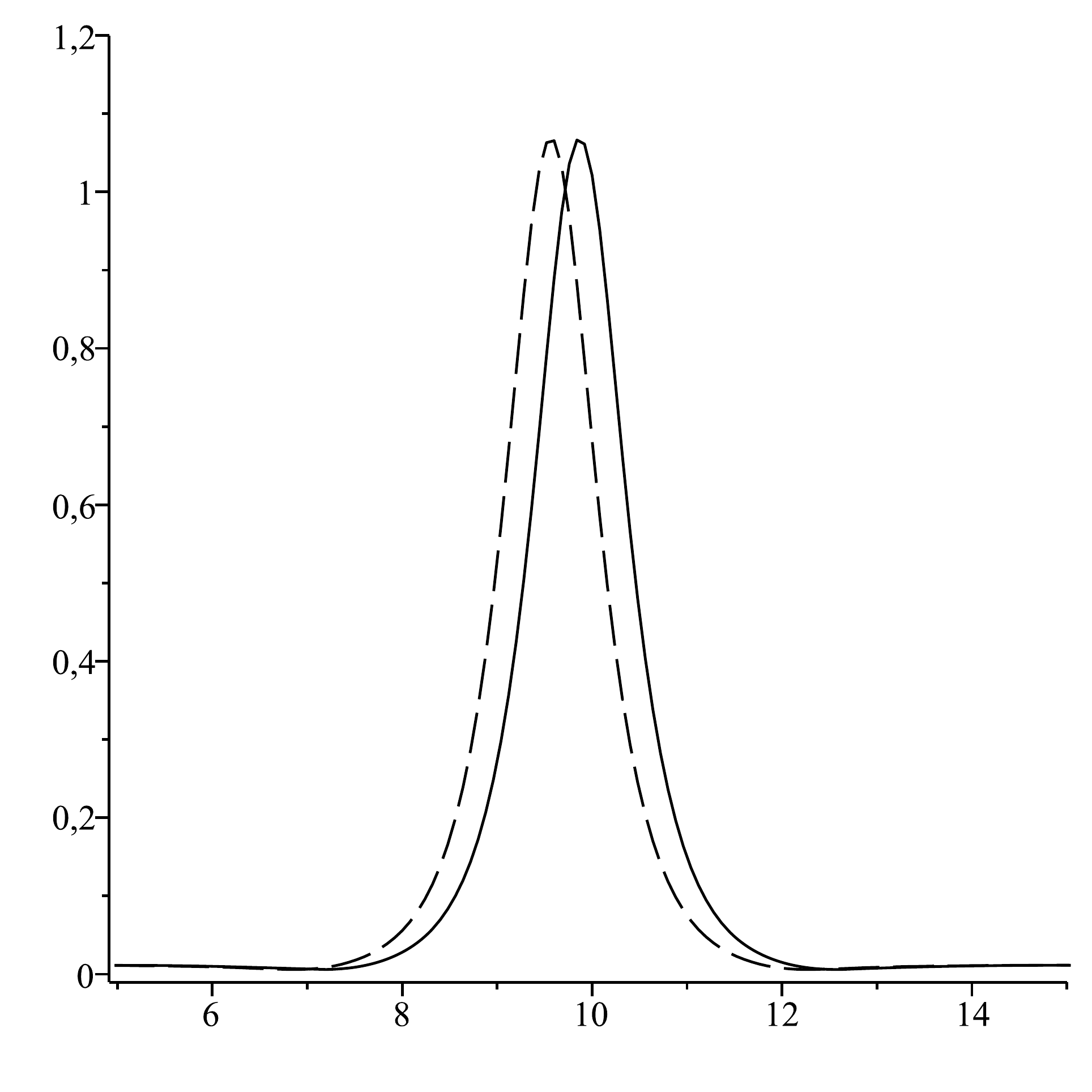}
\caption{Full line is the graph of the function $\rho_2^\infty$, broken line is the graph of the function $\rho_1^\infty$; the two graphs differ because of 
a small translation of the spatial coordinate.}
\label {Fig6}
\end{center}
\end{figure}

\section {Conclusions} \label {Sez5}

Theorem \ref {Teo_uno} states that the result reported in this paper has at least as much theoretical validity as the method based on the standard spectral splitting 
approximation. 

In fact, numerical experiments suggest that this new method has a significantly higher speed of convergence than the standard method and therefore it seems more 
suitable for performing sophisticated numerical experiments. \ {This higher speed of convergence is evident in Table \ref {tabella6} and in Figure \ref {Fig6},  
but in reality it could already be observed in other cases as well, although the effect is much less evident, and it is due to the fact that in our proposed spectral 
splitting method the linear term is exactly treated and not approximated as in the usual approximation. Indeed, if one compares the absolute errors 
$\Delta_1^{n,m}$ and $\Delta_2^{n,m}$ that respectively occur in the usual spectral splitting approximation and in the spectral splitting approximation proposed 
in this paper then it turns out (see Table \ref {tabella7}) that the second approximation proposed here is more efficient in all the situations and, in particular, 
in the case of focusing nonlinearity and inverted oscillator potential. \ The fact that the method proposed in this paper is faster and more accurate than the usual spectral splitting approximation also emerges when comparing the expected values $\langle x \rangle_j^t$, $j=1,2$, of the position observable obtained through the approximate solutions with the exact value $\langle x \rangle^t$ obtained through Ehrenfest's Theorem.}
\begin{table}
\begin{center}
\begin{tabular}{|c|c|c|c||c|} 
\hline 
$V(x)$ & $\nu$ & $n $  &  $m$   &  $ {\Delta_2^{n,m}}/{\Delta_1^{n,m}}$\\ \hline \hline  
$\frac 14 x^2 $      & $+10$ &  $270$   & $9000$   & $0.55$    \\ \hline
$\frac 14 x^2 $      & $-10$ &  $270$   & $9000$   & $0.94$   \\ \hline
$-\frac 14 x^2 $     & $+10$ &  $135$   & $45000$   & $0.66$     \\ \hline
$-\frac 14 x^2 $     & $-10$ &  $135$   & $45000$   & $0.046$     \\ 
\hline
\end{tabular}
\caption{Comparison of absolute errors for different nonlinearities and harmonic/inverted oscillator potentials.}
\label{tabella7}
\end{center}
\end {table}

Not only that, this advantage could become decisive when numerical experiments are performed when the spatial dimension is greater than $1$ and it would be interesting 
to perform a series of experiments to clarify this issue.

On the other hand, the price to pay is due to the fact that the evolution operator associated with the linear Schr\"odinger operator is not always explicitly known; 
however, one could at least partially overcome this defect by using numerical solvers of the Schr\"odinger equation that are sufficiently efficient and fast.

\appendix

\section {Mehler's formula} \label {SezA}

Here we recall the expression for the evolution operator associated to the linear Schr\"odinger operator $H$ with quadratic potential; this expression is named 
{\it Mehler's formula}.

Since the potential is quadratic then the linear operator $H=-\frac {d^2}{dx^2} + \alpha x^2$, $\alpha \in \R$, admits a self-adjoint extension on the 
domain ${\mathcal D}$ and the evolution operator $e^{-iHt}$ is well defined. 

Let $H_0 = -  \frac {\partial^2}{\partial x^2}$ be the free Schr\"odinger operator; then the associated evolution operator has the form 
\bee
\left [ e^{-itH_0 } \psi_0 \right ] (x) = \int_{\R} K_0 (x,y;t) \psi_0 (y) dy \label {eq18Bis}
\eee
where \cite {T}
\bee
K_0 (x,y;t) = \frac {1}{\sqrt {4\pi i t}} e^{i(x-y)^2/4t} \, . \label {M0}
\eee

Let $H_{HO} = -\frac {\partial^2}{\partial x^2} + \frac 14 \omega^2 x^2$, $\omega >0$, be the Harmonic Oscillator Schr\"odinger operator; then the evolution 
operator has the form
\bee
\left [ e^{-itH_{HO} } \psi_0 \right ] (x) = \int_{\R} K_{HO} (x,y;t) \psi_0 (y) dy \label {eq19Bis}
\eee
where \cite {F}
\bee
K_{HO} (x,y;t) = \sqrt {\frac {\omega}{4\pi i \sin (\omega t)}} \exp \left \{ i\frac {\omega}{4\sin (\omega t)} \left [ (x^2+y^2) \cos (\omega t) - 2 xy \right ] 
\right \} \, . \label {MHO}
\eee

Let $H_{IO} = -\frac {\partial^2}{\partial x^2} - \frac 14 \omega^2 x^2$, $\omega >0$, be the Inverted Oscillator Schr\"odinger operator; then the evolution 
operator has the form
\bee 
\left [ e^{-itH_{IO} } \psi_0 \right ] (x) = \int_{\R} K_{HO} (x,y;t) \psi_0 (y) dy \label {MIO}
\eee
where \cite {BJM,B}
\be
K_{IO} (x,y;t) 
= \sqrt {\frac {\omega}{4\pi i \sinh (\omega t)}} \exp \left \{ i\frac {\omega}{4\sinh (\omega t)} \left [ (x^2+y^2) \cosh (\omega t) - 2 xy \right ] \right \}\, . 
\ee

\begin {remark}
It is well known that
\be
\| e^{-iH\delta } \psi_0 \|_{L^2} = \|  \psi_0 \|_{L^2} 
\ee
for any self-adjoint operator $H$. \ Furthermore, in the case of self-adjoint operator $H$ with quadratic potential then from (\ref {M0}), (\ref {MHO}) and 
(\ref {MIO}) it follows that 
\be
\| e^{-iHt } \psi_0 \|_{L^\infty} \le C t^{-1/2} \|  \psi_0 \|_{L^1} 
\ee
for any $\alpha =\pm \frac 14 \omega^2  \in \R$ and for any $t \le t^\star $, where $ t^\star  < 
\frac {\pi}{\omega} $ is fixed, and for some $C=C(t^\star , \omega)$. 
\end {remark}

\end {document}